\documentclass[letterpaper]{article}\usepackage{aaai2027}
\nocopyright
\usepackage{natbib}
\usepackage[hyphens]{url}
\usepackage[utf8]{inputenc}
\usepackage{graphicx}
\usepackage{caption}
\usepackage{booktabs}
\usepackage{amsmath,amssymb,amsthm}
\usepackage{multirow}
\usepackage{listings}
\usepackage{array}
\title{Local Agreement Does Not Imply Global Consensus in LLM Panels}
\author{Machiko Hirota\textsuperscript{1}, Akshara Nadayanur Sathis Kanna\textsuperscript{2}\equalcontrib, Ujwal Kumar\textsuperscript{3}\equalcontrib, Phan Xuan Tan\textsuperscript{4}}
\affiliations{\textsuperscript{1}Shiba AI\\\textsuperscript{2}Carnegie Mellon University\\
\textsuperscript{3}University of Tokyo\\\textsuperscript{4}Shibaura Institute of Technology}
\begin{document}
\maketitle
\begin{abstract}
Local coordination in an LLM panel need not produce a common answer. We study 50-agent panels that update clock-time estimates after observing their neighbours. This circular response space separates agreement across the whole panel from agreement within neighbourhoods. Increasing the reasoning effort of gpt-5-mini reduces variation in local agreement, but all 20 higher-effort trials retain a one-twist configuration: neighbouring estimates remain close while the panel spans the clock. Degree-preserving rewiring instead produces global synchronisation in all 25 primary clock controls. These classifications are unchanged across 2,187 combinations of diagnostic thresholds and reporting windows. A radius-, effort-, and passage-matched scoring task also shows more uniform local agreement after rewiring, without establishing greater scoring accuracy. A simple averaging model helps interpret the clock results but predicts their configurations poorly. The experiments support measuring global agreement and local structure separately, rather than inferring consensus from uniform local agreement.
\end{abstract}

\section{Introduction}
\label{sec:intro}
Multi-agent language-model systems exchange responses to support evaluation, debate, and collaborative problem solving \citep{chan2023chateval,verga2024replacing,du2023improving,liang2023encouraging,li2023camel,wu2023autogen}. We examine a separate question: when agents appear to coordinate, do they agree across the whole panel, or only within local groups?

Consider estimates of a time of day. All agents may report the same time. Alternatively, neighbouring agents may report similar times while their estimates gradually traverse the full clock around a ring. The second configuration has strong local order but no common panel-wide estimate. A third possibility is that an ordered region coexists with a region whose estimates remain dispersed. An average of local agreement, or the variation in that agreement, does not distinguish all three cases.

We examine these configurations in a deliberately restricted protocol. Fifty agents make synchronous updates from their own previous estimate and the estimates of a fixed set of peers. The clock task gives the answers a known geometry; it does not require an embedding or a learned similarity function. We vary reasoning effort and the communication graph while keeping the task and panel size fixed. A small extension uses scalar scores for three passages. The protocol isolates response propagation from the other operations in an open-ended multi-agent application.

Our terminology comes from coupled dynamical systems. A \emph{twisted} configuration winds around a circular response space, whereas a \emph{chimera} in oscillator systems contains coherent and incoherent populations \citep{wiley2006size,kuramoto2002,abrams2004chimera}. We use \emph{chimera-like} for an operational pattern in the observed LLM outputs, not as evidence that the agents obey oscillator equations or possess an asymptotically stable chimera state.

Reasoning effort and topology affect different aspects of coordination. Higher-effort gpt-5-mini panels become more locally ordered without reaching global consensus; degree-preserving rewiring produces globally synchronised clock-task panels. Direct winding measurements and threshold/window checks support this distinction. We first define the protocol and measurements, then report the effort, topology, and radius comparisons. A short interpretation follows the results. The contribution is empirical: standard averaging theory explains one possible source of ordered disagreement, but does not establish the behaviour of the LLM panels.

\section{Tasks and experimental design}
\label{sec:method}
\subsection{Panel protocol}
\label{sec:protocol}
Let $G=(V,E)$ be a communication graph with agents indexed by $i\in\{0,\ldots,N-1\}$, where $N=50$. Agent $i$ observes the agents in $\mathcal N_i$. At turn zero, it receives the task and a private prior and returns an initial estimate. At each later turn, it receives the task, its own previous estimate, and its neighbours' previous estimates. All updates use the same preceding turn. Each request is stateless: there is no conversation history beyond the values placed in the current prompt. The clock prompt requests a short justification and a machine-readable answer; the peer information passed between turns consists of estimates, not an exchange of those justifications.

Requests use temperature zero through OpenRouter. This is a decoding setting, not a guarantee that repeated API calls produce identical answers. Failed requests are retried; persistent failures carry forward the previous valid estimate and are marked as fallbacks. A fallback is not a fresh model response. The primary analysis contains 340 LLM trials (228 clock and 112 judging), with 20 clock and 30 judging calibration trials. Twelve additional clock trials test permuted initialisation. Most trials run to $T=30$; the lifetime experiments extend selected conditions to $T=60$ or $T=120$. The judging analysis excludes 27 failed or invalid originals and includes four valid replacements. Appendix~\ref{app:trial-counts} gives counts for each experimental block and the exclusion criteria.

\begin{figure*}[!t]
  \centering
  \includegraphics[width=\textwidth]{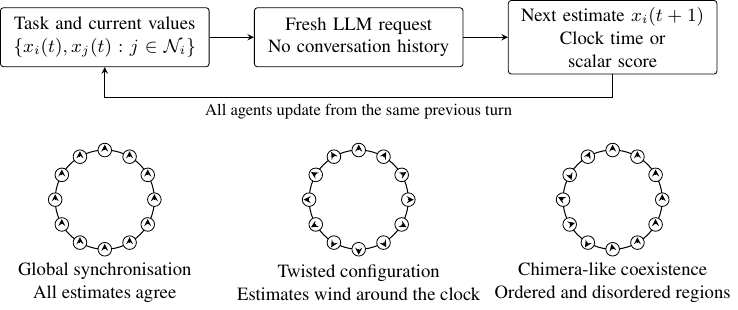}
  \caption{Panel protocol and illustrative clock configurations. The graph determines which previous estimates enter each request. The lower diagrams are schematics, not experimental observations: each small hand denotes an agent's clock estimate. Local agreement can be uniform in both a synchronised panel and a twisted panel, although only the former has a common global estimate. The experiments use 50 agents; fewer are drawn for clarity.}
  \label{fig:protocol}
\end{figure*}

\subsection{Response spaces and instructions}
\label{sec:tasks}
\paragraph{Clock task.}
Agents act as event analysts and return a peak time in HH:MM format. If the parsed answer is $u_i(t)\in\{0,\ldots,1439\}$ minutes, its phase is
\begin{equation}
\theta_i(t)=\frac{2\pi u_i(t)}{1440}\in[0,2\pi).
\label{eq:phase}
\end{equation}
This encoding respects the task: 23:59 and 00:00 are one minute apart. The initial priors are $\tau_i=\operatorname{round}(1439i/(N-1))$ minutes in agent-index order. They approximate a one-twist configuration but are not the exact equally spaced phases used in the analytical calculation below. The follow-up permutes the same priors across agents, preserving their values while changing their spatial arrangement.

From turn one, the prompt instructs agents to give equal weight to their own previous estimate and the collective signal from their neighbours, using circular arithmetic. The LLM interprets this instruction; the experimental system does not replace its response with a programmed average. This makes instruction following and deviations from it observable. The task has no externally verified event time, so agreement is not a measure of forecasting accuracy. The verbatim clock prompts appear in Appendix~\ref{app:prompts}.

\paragraph{Judging task.}
Agents score three short passages about a lighthouse on a scale from 1 to 100. The three passages are labelled \texttt{strong}, \texttt{weak}, and \texttt{ambiguous}; these are design labels, not independently validated quality judgments. The prompts supply neither scoring anchors nor a rubric. Initial impressions are $\tau_i=1+\operatorname{round}(99i/(N-1))$ and appear only at turn zero. Later prompts contain the passage, the agent's previous score, and its peers' scores, without the clock task's equal-weight instruction. The response space is an interval: 1 and 100 are far apart. We analyse this as a separate score-propagation task, not a test of identical dynamics under a different encoding. The full passages and prompts appear in Appendix~\ref{app:judging-materials}.

\subsection{Communication graphs and model conditions}
\label{sec:topologies}
On a radius-$R$ ring, $\mathcal N_i=\{i\pm1,\ldots,i\pm R\}\pmod N$, so each agent sees $d=2R$ peers. ``Nearest'' refers to this index ordering, not semantic similarity or a physical location. The radius sweep uses $R\in\{3,4,5,6,7,9,10\}$ with Qwen-2.5-72b and gpt-5-nano, at three trials per cell. The primary gpt-5-mini effort comparison uses $R\in\{5,7\}$, with eight minimal-effort and five trials at each higher effort per radius.

The primary control randomises the ring through symmetric double-edge swaps, preserving every agent's degree and avoiding self-loops or duplicate edges. It holds peer count and visibility symmetry fixed, not the resulting text, token count, or effective influence weights. We call this \emph{degree-preserving rewiring}; finite swapping does not establish uniform sampling from all regular graphs. The target is 20 accepted swaps per original edge, without enforcing connectivity during construction; all analysed graphs were connected.

Watts--Strogatz-style (WS) graphs at rewiring probabilities $p\in\{0.1,0.5\}$ provide additional graph instances for the lifetime comparison. They preserve mean degree, not each agent's degree, and are not interchangeable with the degree-preserving control. A secondary directed-random control fixes only out-degree and is not used as the primary matched-topology intervention. Appendix~\ref{app:graph-construction} specifies the graph algorithms. In total, the study covers 14 model/effort conditions across six providers. The principal controlled results concern gpt-5-mini; Qwen and Gemini supply further topology comparisons. Gemini effort comparisons are descriptive because provider routing and the pooled radius sets differ between settings.

\section{Measuring collective organisation}
\label{sec:measurement}
\label{sec:regimes}
\subsection{Global order, local order, and spatial heterogeneity}
For clock estimates, define
\begin{equation}
\begin{aligned}
 r(t)&=\left|\frac1N\sum_{j=0}^{N-1}e^{\mathrm i\theta_j(t)}\right|,\\
 Z_i(t)&=\left|\frac{1}{|\mathcal N_i\cup\{i\}|}
       \sum_{j\in\mathcal N_i\cup\{i\}}e^{\mathrm i\theta_j(t)}\right|.
\end{aligned}
\label{eq:order}
\end{equation}
Both lie in $[0,1]$. Global order $r=1$ means all phases agree. A small $r$ means their directions cancel; it does not distinguish random answers from an orderly winding. Local order $Z_i$ measures concentration in agent $i$'s self-inclusive neighbourhood. The same construction is used throughout the clock analysis.

For the reporting window $W=\{20,\ldots,T\}$, define
\begin{equation}
\begin{aligned}
 \bar r&=\frac{1}{|W_r|}\sum_{t\in W_r}r(t), &
 \bar Z_i&=\frac{1}{|W_i|}\sum_{t\in W_i}Z_i(t),\\
 \bar Z&=\frac1N\sum_i\bar Z_i, &
 \Delta Z&=\max_i\bar Z_i-\min_i\bar Z_i.
\end{aligned}
\label{eq:window}
\end{equation}
Here $W_i$ contains turns in $W$ with a fresh, parseable response from agent $i$, and $W_r$ contains turns with at least one such response. Stored order parameters are calculated from the complete panel state, including carried-forward values. Omitting a focal fallback row therefore does not remove that value from its peers' neighbourhoods. Retaining all recorded states instead changes no primary clock labels.

We call $W$ a reporting window rather than assume stationarity. At $T=30$, it contains 11 turns before any response filtering. $\Delta Z$ measures variation in local agreement, not disagreement between estimates. An exactly synchronised panel and an exact one-twist ring both have $\Delta Z=0$; the latter has $\bar r=0$. 

\subsection{Operational labels and collapse time}
\label{sec:classification}
Let $V_t$ contain agents with a fresh response at turn $t$, and define $h(t)=\max_{i\in V_t}Z_i(t)-\min_{i\in V_t}Z_i(t)$. The collapse onset is
\begin{equation}
\begin{aligned}
 t_{\mathrm{col}}&=\min\{t:h(t+s)<0.03\ \forall s\in\{0,\ldots,4\}\}.
\end{aligned}
\label{eq:collapse}
\end{equation}
The search begins at turn zero and requires five consecutive observed turns. Missing turns interrupt the interval. The event is confirmed at $t_{\mathrm{col}}+4$; without confirmation by $T$, onset is conservatively censored beyond $T-4$. Collapse does not imply global consensus or rule out renewed heterogeneity.

Table~\ref{tab:classifier} gives the clock classifier, applied in priority order to reporting-window summaries of both observed and simulated trials. The rules are descriptive, not preregistered phase boundaries. The residual label need not mean independent random behaviour.
\begin{table}[t]
\centering\small
\begin{tabular}{@{}p{0.26\linewidth}p{0.69\linewidth}@{}}
\toprule
Label & Rule, in priority order\\
\midrule
Synchronised & $\bar r>0.9$.\\
Twisted & $\bar r<0.1$ and $\bar Z>0.7$.\\
Chimera-like & No detected collapse, $\max_i\bar Z_i>0.9$, and $\min_i\bar Z_i<0.7$.\\
Other / incoherent & None of the preceding rules applies.\\
\bottomrule
\end{tabular}
\caption{Operational clock-task labels. The twisted rule detects low global and high local order but does not itself measure winding. The chimera-like rule does not prove temporal phase locking or classical chimera dynamics.}
\label{tab:classifier}
\end{table}
A separate audit finds contiguous low-order domains of at least two agents in 53 of the 55 chimera-like trials; the median longest domain is 10 agents. Contiguity was not used to relabel trials. We check winding separately by summing the signed shortest angular steps around the ring. A total of $2\pi$ is one positive winding; zero global order alone does not establish a winding. Section~\ref{sec:interventions} reports the winding and threshold-sensitivity results.

\subsection{Calibration and the scalar extension}
\label{sec:baselines}
At $R\in\{5,7\}$, uniform-random and random-neighbour-copying agents run under the same clock protocol, with five trials per cell. Their mean $\Delta Z$ values are $(0.162,0.140)$ and $(0.414,0.247)$, respectively. The corresponding cell means of $\max_i\bar Z_i$ are at most 0.346 and 0.835. Thus, appreciable heterogeneity can arise without an LLM and without the high-order region required by the chimera-like criterion. These controls calibrate interpretation; they do not establish that a particular LLM pattern has a uniquely language-model mechanism.

For judging scores $s_i(t)\in[1,100]$, use a separate notation,
\begin{equation}
 A_i(t)=\operatorname{clip}_{[0,1]}\left(1-\frac{2\,\operatorname{sd}
 \{s_j(t):j\in\mathcal N_i\cup\{i\}\}}{99}\right),
\label{eq:scalar-order}
\end{equation}
with $a(t)$ defined over all agents. Window averages give $\bar A_i$ and $\Delta A=\max_i\bar A_i-\min_i\bar A_i$. The scalar collapse threshold is 0.10. Here $\operatorname{sd}$ is the population standard deviation (divisor equal to the number of scores). The normalisation expresses dispersion relative to the score range; it does not calibrate $A_i$ to circular order $Z_i$. For example, $a>0.9$ means a panel standard deviation below 4.95 score points, not a circular concentration of the same strength. We compare effects within each task and do not interpret scalar configurations as phase twists. The two protocols also differ in their update instructions, so the extension does not isolate response geometry alone.

\subsection{Statistical analysis}
Matched contrasts average stratum-specific differences equally: radius strata for clock effort and passage strata for judging. We bootstrap trials within each condition--stratum cell (20,000 resamples) and permute condition labels within strata, using exact enumeration when feasible and otherwise 100,000 draws with a plus-one correction. Holm correction is applied separately to the three clock-effort, three judging-effort, and two judging-topology contrasts. These are post hoc reanalyses, not preregistered tests. The matched Qwen comparison uses its single common radius.

Trials, not agent-turns, are the resampling units. Inference is conditional on the sampled passages and graph instances; three passages cannot establish task-population transfer. Intervals are unadjusted percentile intervals, not inversions of the corrected permutation tests. Table~\ref{tab:matched-contrasts} reports all pairwise contrasts in these post hoc analyses; the primary effects and sample sizes appear below.

\section{Results}
\label{sec:results}
\subsection{Higher reasoning effort increases local order, not global consensus}
\label{sec:interventions}
In the \texttt{gpt-5-mini} clock experiments, the radius-stratified minimal-minus-low difference in $\Delta Z$ is 0.4191 (95\% CI [0.3510, 0.4850]); minimal minus medium is 0.4203 [0.3522, 0.4863]. Both have Holm-adjusted $p<10^{-4}$ (Figure~\ref{fig:effort}). All 20 low- and medium-effort trials are classified as twisted, not synchronised. Their mean global orders are 0.010 and 0.008, respectively. Direct calculation confirms winding number $+1$ at every reporting-window turn in all 20 trials. These panels are locally orderly but remain divided across the clock. The decrease in $\Delta Z$ therefore reflects more uniform local order, rather than global consensus.

\begin{figure}[t]
  \centering
  \includegraphics[width=\linewidth]{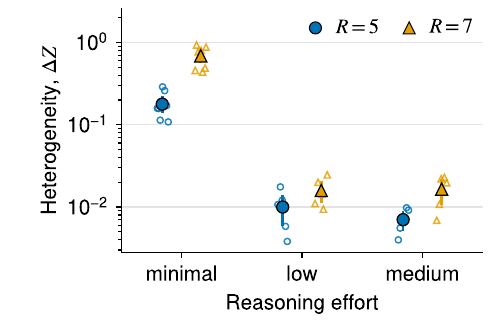}
  \caption{Clock-task reasoning-effort comparison for gpt-5-mini. Hollow markers show trials; filled markers and bars show cell means and bootstrap 95\% intervals, on a logarithmic scale at $R=5$ and $R=7$. There are eight minimal-effort and five low- or medium-effort trials per radius. All higher-effort trials satisfy the twisted criterion; their low heterogeneity is not global synchronisation.}
  \label{fig:effort}
\end{figure}

The ordered prior assignment already has low global and high local order: $r(0)\approx0.019$, with mean local order approximately 0.923 at $R=5$ and 0.859 at $R=7$. Higher effort may therefore preserve an ordered starting state rather than create it. The $R=7$ follow-up distinguishes these possibilities only partially. Under permuted priors, both medium-effort ring trials ended twisted, with opposite winding directions. None of the five minimal-effort ring trials did: two were chimera-like, two synchronised, and one fell in the residual category. All five permuted minimal-effort rewired trials synchronised. A twisted outcome can therefore form after spatial permutation, but two medium-effort runs do not establish a robust formation rate or eliminate initialisation dependence.

Across 2,187 combinations of thresholds and reporting windows, the 20 higher-effort trials remain twisted and all 25 primary degree-preserving clock trials remain synchronised. The number labelled chimera-like varies from 41 to 68 of 228 trials; these labels are more sensitive than the principal contrasts.

\subsection{Degree-preserving rewiring produces global agreement}
\label{sec:topology-results}
The clock-task rewired controls synchronise, whereas the corresponding ring conditions can remain spatially heterogeneous (Figure~\ref{fig:topology}). Unlike the reasoning-effort result, this intervention changes global agreement as well as heterogeneity. Each node retains its peer count; the intervention changes which peers it sees.

\begin{figure*}[!t]
  \centering
  \includegraphics[width=0.95\textwidth]{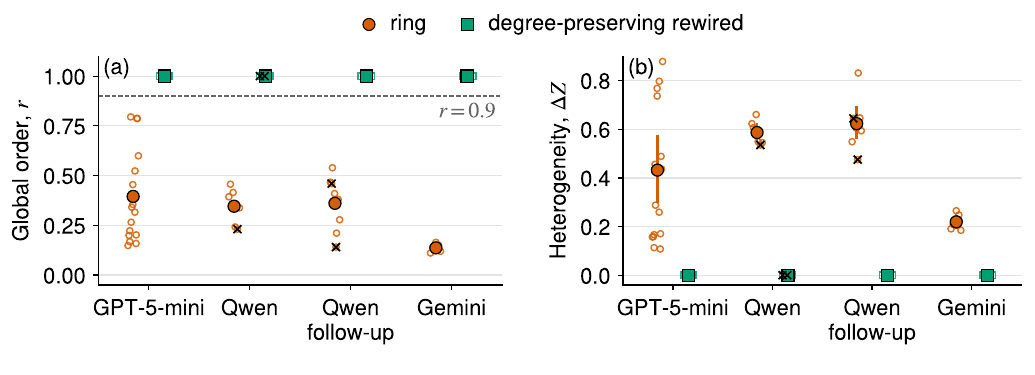}
  \caption{Clock-task topology comparisons. Hollow markers show trials and filled markers show means; vertical bars in the right panel are bootstrap 95\% intervals. Crosses mark retained trials with more than 5\% fallback responses. Global order (left) and heterogeneity (right) distinguish the topologies. Some displayed summaries pool radii; they are descriptive. In the Qwen follow-up, the matched inferential comparison uses $R=7$ only, because no rewired $R=5$ arm is available.}
  \label{fig:topology}
\end{figure*}

The Qwen follow-up illustrates why matching matters. At $R=7$, ring trials have mean $\Delta Z=0.637$ ($n=4$), compared with 0.000 for the rewired control ($n=3$). The difference is 0.637 with a 95\% interval of $[0.518,0.772]$ and an exact permutation $p=0.0286$. Pooling additional $R=5$ ring trials yields a similar descriptive mean, but changes the comparison as well as its sample size; it is not a second matched intervention.

Gemini provides a further topology comparison. Its effort comparison is descriptive because provider routing and radius composition differ (Appendix~\ref{app:gemini}).

\subsection{Rewiring also reduces variation in passage scores}
\label{sec:judging-results}
At $R=7$ and minimal effort, the ring-minus-rewired difference in $\Delta A$ is 0.1293 (95\% CI [0.1231, 0.1357]; Holm-adjusted $p<10^{-4}$). This comparison uses 15 trials per topology, five for each passage. It holds radius, effort, and passage composition fixed. The WS-style control gives a difference of 0.1033 [0.0965, 0.1101], but does not preserve individual degrees. Each rewired judging condition uses one graph instance, so these tests do not estimate an effect over independently sampled graphs.

\begin{figure*}[!t]
\centering
\includegraphics[width=0.95\textwidth]{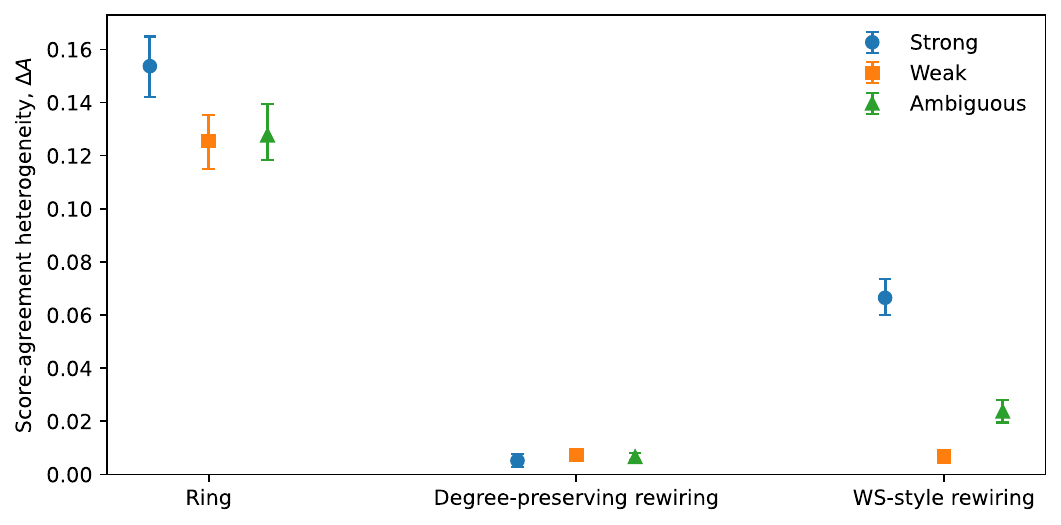}
\caption{Matched judging-task comparison at $R=7$ and minimal effort. Points are passage-level means with percentile-bootstrap 95\% intervals (five trials per passage and topology). Labels denote experimental passages, not validated quality categories. Each rewired condition uses one fixed graph.}
\label{fig:judging-matched}
\end{figure*}

For ring panels at $R=7$, the passage-stratified minimal-minus-medium effort difference is 0.0367 [0.0229, 0.0494] (Holm-adjusted $p=0.0011$). The other two pairwise effort comparisons do not survive correction. The evidence concerns agreement among scores, not their quality or calibration to human judgments. Table~\ref{tab:matched-contrasts} reports all pairwise contrasts; Figure~\ref{fig:judging-matched} displays the matched passage-level means.

\subsection{Radius changes the pattern of disagreement}
\label{sec:phase-diagram}
\label{sec:taxonomy}
Both clock radius sweeps attain their largest observed mean $\Delta Z$ at $R=6$: 0.725 for Qwen and 0.573 for gpt-5-nano. The high-radius behaviour differs. At $R=10$, Qwen still has mean $\Delta Z=0.484$, while gpt-5-nano has 0.190 and greater global order. With three trials per cell, these are exploratory profiles, not an established universal bell-shaped phase diagram. Across the model conditions, the operational classifier assigns synchronised, twisted, chimera-like, and residual outcomes. Figure~\ref{fig:phase-diagram} shows both measurements across the radius sweep.

\begin{figure*}[!t]
\centering
\includegraphics[width=0.95\textwidth]{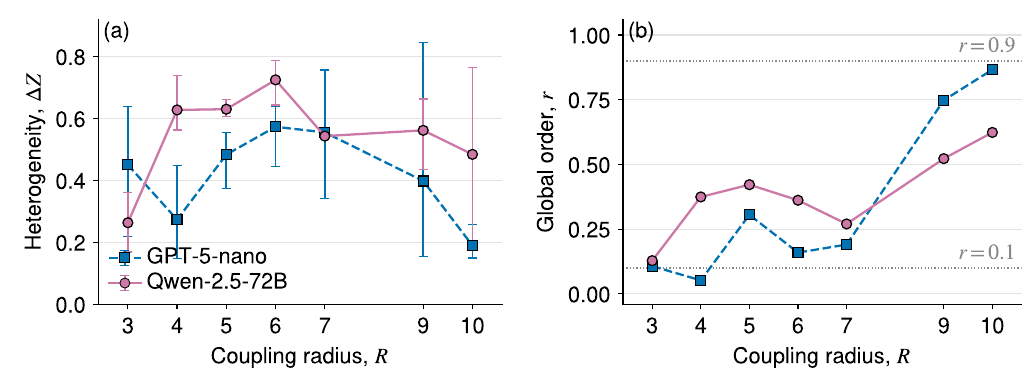}
\caption{Clock-task radius sweep ($n=3$ per radius). Markers show means; left-panel bars are bootstrap 95\% intervals. Heterogeneity and global order change with radius, but the high-radius responses differ across Qwen and gpt-5-nano. These are exploratory profiles, not established phase boundaries.}
\label{fig:phase-diagram}
\end{figure*}

Low heterogeneity also remains insufficient. Of the 90 clock trials with $\Delta Z<0.03$, 36 (40.0\%) have $\bar r<0.9$ and are labelled twisted. This is a failure of a \emph{heterogeneity-only} convergence criterion: global order would reveal the disagreement. Reporting both quantities avoids that error.
\label{sec:practical}

\section{Why local coordination can persist}
\label{sec:model}
Consider an exactly even spacing of clock estimates around a ring. An agent has equally many neighbours on either side of its own estimate. At the radii studied here, their circular mean points to its own time. Moving toward that mean therefore leaves the configuration unchanged, even though the panel has no common answer. This is a fixed-point property of noiseless circular averaging, not a new synchronisation result or a guarantee that LLM agents converge to it \citep{wiley2006size}. It explains why closer adherence to the clock prompt can preserve local order without producing global agreement.

We assess that interpretation by regressing each agent's angular movement on the shortest angular displacement toward the mean of its neighbours. A response coefficient of $\beta=0.5$ is the local reference for the prompt's half-own, half-peer instruction. The fitted coefficients for minimal, low, and medium effort are 0.221, 0.468, and 0.489; the corresponding residual standard deviations are 0.287, 0.043, and 0.048 radians. Higher-effort updates are thus closer to the instructed average. The minimal-effort regression has $R^2=0.089$, so averaging explains little of its update variation. These fits describe responses, not internal reasoning mechanisms.

A stochastic approximation adds Gaussian perturbations and occasional copying of a neighbour's estimate. The copying rate is estimated from changed answers that match a visible peer, which can also arise through rounding or coincidence. It is distinct from carrying forward an agent's own answer after an API failure. With parameters fitted to the observed conditions, this model matches 11 of 22 condition labels, compared with 10 for always predicting twisted. It matches 9/10 twisted and 2/6 chimera-like cells, but no synchronised cells. It is not a validated predictor, and the empirical conclusions do not depend on it. Appendix~\ref{app:model} specifies the fitted rule and its limitations.

\section{Relation to prior work}
\label{sec:related}
Multi-agent debate and evaluation methods ask whether interaction improves reasoning or judgment \citep{du2023improving,chan2023chateval,verga2024replacing}. Comparisons with strong single-agent and ensemble methods show that additional interaction does not consistently improve task performance \citep{wang2024rethinking,smit2024mad}. Our experiment addresses a different outcome: the spatial distribution of agreement under a controlled update protocol. It neither establishes a performance advantage for panels nor measures an accuracy cost of fragmentation.

Work on networked LLMs already connects communication structure to convergence. \citet{kayaalp2026llm} study textual opinion exchange through a DeGroot framework and relate transient convergence to network spectra, while finding that limiting opinions depart from the linear model's prediction. \citet{han2026conformity} study topology and self--social weighting in misinformation judgments, including agreement on incorrect answers. The contribution here is not the general observation that graph structure affects LLM consensus. It is the distinction between locally ordered non-consensus and global agreement, examined through effort and degree-preserving interventions in a response space with explicit geometry.

Sycophancy and conformity research studies how models defer to external or peer signals \citep{sharma2024towards,baltaji2024conformity,pitre-etal-2025-consensagent}. Our copying proxy and fitted response coefficient describe output transitions, not the psychological or internal computational causes of deference. Similarly, a provider's reasoning-effort setting is not a direct measurement of how an agent reasons.

Twisted and chimera states supply the diagnostic vocabulary \citep{kuramoto2002,abrams2004chimera,wiley2006size,panaggio2015chimera}. Standard consensus theory supplies the spectral comparison \citep{xiao2004fast,olfatisaber2004consensus}. This vocabulary does not establish dynamical equivalence: even oscillator chimeras can be transient \citep{wolfrum2011chimera}.

\section{Discussion and limitations}
\label{sec:discussion}
Uniform local agreement and a common panel-wide answer are different outcomes. In this protocol, greater reasoning effort and degree-preserving rewiring produce the former and the latter, respectively. A panel assessment should therefore report global agreement alongside local-order profiles and the observation window. This avoids the demonstrated error of calling low heterogeneity consensus. It does not establish added predictive value over a sensible global-agreement measure.

The clock task isolates response propagation under an averaging instruction. Agents exchange estimates, not substantive arguments, and ordinary numerical systems can also sustain ordered disagreement. The three judging passages extend the response domain but lack validated quality labels and scoring anchors. Neither task tests whether agreement improves answer quality; rewiring is not established as a way to improve deployed reasoning systems.

The study uses one main panel size, a narrow graph family, and priors that favour a winding. The two higher-effort permuted trials demonstrate that winding can form, not a robust formation rate. Several descriptive cells are small, and repeated trials on one graph do not establish generalisation across graphs. Provider aliases, routing differences, and failed requests further limit exact replication. The experimental code and raw trajectories are not provided with this submission; the protocols, prompts, and analysis definitions are documented in the paper.

\section{Conclusion}
\label{sec:conclusion}
Local order, uniformity of local order, and global consensus are distinct properties. In the tested protocol, higher reasoning effort produces locally ordered non-consensus, whereas degree-preserving rewiring produces global agreement. Joint measurements distinguish these outcomes. The fitted model remains descriptive; task-level validation is needed before drawing conclusions about reasoning quality.

\section*{AI Use}
AI was used in assisted with manuscript revision, coding assistance in implementing experiments, generating figures, and mathematical checks. 

\begingroup

\bibliography{custom}
\endgroup

\appendix
\section{Experimental details and robustness}
\label{app:methods}
The following sections specify the samples, graph algorithms, stimuli, and analysis.

\subsection{Samples and model conditions}
\label{app:trial-counts}
Table~\ref{tab:samples} accounts for the 228 primary clock trials. In addition, the study has 112 retained judging trials, 12 permuted-prior clock trials, and 50 calibration trials. Standard trials use $N=50$, $T=30$, and temperature zero. The 40 extended clock trials use $T=60$ except ring radii 6 and 7, which use $T=120$. The higher-effort formation test contains two medium-effort ring trials; the permuted minimal-effort ring and rewired arms contain five each.

\begin{table*}[t]
\centering\small
\begin{tabular}{@{}p{0.22\textwidth}p{0.67\textwidth}r@{}}
\toprule
Clock experiment & Conditions and counts & Total\\
\midrule
Effort & GPT-5-mini: $R=5,7$, eight minimal and five each low/medium per radius. GPT-oss-120b: $R=5,7$, five minimal and three default per radius. & 52\\
Radius & Qwen-2.5-72b and GPT-5-nano: $R=3,4,5,6,7,9,10$, three each. DeepSeek: $R=(1,2,5,6,7,8,9)$, $n=(2,1,5,2,2,2,2)$. & 58\\
Other model comparisons & Llama-4-scout and Mistral-large-2512: five at each of $R=5,7$. Mistral-large-2411: two at each radius. & 24\\
Topology & GPT-5-mini and Qwen: three rewired trials per radius ($R=5,7$). Directed controls: two GPT-5-mini at $R=5$; two Qwen per radius. & 18\\
Qwen follow-up & Four ring trials per radius ($R=5,7$) and three rewired trials at $R=7$. & 11\\
Gemini & Five per condition: default ring $R=5,7$; low and medium ring $R=7$; default rewired $R=7$. & 25\\
Extended duration & Five per graph: rings $R=5,6,7$; two WS-style graphs each at $p=0.1,0.5$ ($R=7$); one degree-preserving graph ($R=7$). & 40\\
\midrule
Total & Primary clock dataset; excludes permuted-prior and calibration trials. & 228\\
\bottomrule
\end{tabular}
\caption{Primary clock sample counts. Unspecified effort is the provider default, except the GPT-5-nano radius sweep and GPT-5-mini topology/duration runs, which use minimal effort. Each primary degree-preserving run belongs to one listed block; their combined count is 25.}
\label{tab:samples}
\end{table*}

All judging trials use GPT-5-mini. Ring minimal-effort cells have five trials per passage at $R=5$ and $R=7$; at $R=6$, strong has seven and the other passages five each. At $R=7$, low effort has five per passage; medium has five strong, eight weak, and seven ambiguous. The degree-preserving and WS-style $p=0.5$ controls each have five per passage at $R=7$, minimal effort. These counts sum to 112. The numerical calibration agents output uniformly random values or copy a random neighbour; each has five trials per clock radius ($R=5,7$) and judging radius ($R=5,6,7$).

The 14 clock model/effort settings comprise GPT-5-mini (minimal/low/medium), GPT-5-nano (minimal), GPT-oss-120b (minimal/default), Gemini-3.1-flash-lite (default/low/medium), and Qwen-2.5-72b, Llama-4-scout, DeepSeek-V4-flash, Mistral-large-2411 and Mistral-large-2512 at default effort. Requested aliases and returned providers were recorded, but resolved backend snapshot IDs were not. ``Default'' denotes an unset effort parameter, not a common reasoning budget across models.

\subsection{Failures, parsing, and replacements}
\label{app:trial-validity}
Failed or unparsable requests are retried. Persistent failure carries forward the previous estimate, or the prior at turn zero. The clock parser searches trailing JSON objects for a valid \texttt{peak\_time}, then accepts a valid HH:MM substring. The judging parser accepts the first standalone integer from 1 to 100, even if extra text violates the output instruction.

Of 135 original judging LLM trials, 23 produced no usable observations. Four further trials were excluded for at least five consecutive turns with no valid responses, no valid turn-zero response, or no valid response at $t\geq20$. Each of those four has one valid replacement, giving 108 valid originals plus four replacements. Replacements retained task conditions, prompts, model alias, and priors; two used a 180-second rather than 60-second timeout, and execution was sequential rather than batched. All 112 retained trials have zero fallbacks. Identical backend versions cannot be verified behind the aliases.

Included clock conditions have mean fallback rates from zero to 3.18\% (Qwen); one Qwen trial has 37.87\%. Primary order summaries omit the focal fallback rows, not the carried-forward values used by neighbouring agents. Retaining all states leaves every primary clock label unchanged. These checks describe the observed failures, not the unobserved responses that successful requests might have produced.

\subsection{Graph construction}
\label{app:graph-construction}\label{app:settings}
The ring connects $i$ to $i\pm1,\ldots,i\pm R$ modulo $N$. For degree-preserving rewiring, edges are stored as ordered endpoint pairs. Two edge-list positions are sampled; the first edge is reversed with probability $1/2$; and $(a,b),(c,d)$ is replaced by $(a,d),(c,b)$. A proposal requires four distinct endpoints and no existing replacement edge. The target is $20|E|$ accepted swaps, with at most $200$ times that target in attempts. Connectivity is not imposed; all analysed undirected graphs were connected. Finite swapping is not assumed to sample uniformly from regular graphs.

The WS-style construction visits the sorted original edge list. For $(u,v)$ with $u<v$, it keeps $u$ and, with probability $p$, replaces $v$ by a uniformly chosen non-neighbour other than $u$. It preserves mean degree, not each degree. The directed control independently chooses $2R$ peers per agent. The spectral diagnostic uses the undirected union: an edge is present when either direction is visible; this is not the directed update operator.

A supplied graph seed is used directly; otherwise it is $\texttt{base\_seed}+1000\,\texttt{trial\_idx}+R$. Graphs and priors are locally seeded, but no sampling seed is sent to the model API. Reconstructed neighbour lists were checked against the logged turn-one prompts. The regular and WS-style judging arms each retain one fixed graph across the three passages. Ring graphs are fixed by $N$ and $R$.

\subsection{Clock-task prompts}
\label{app:prompts}

The following templates were used for the clock task
(\texttt{v1.0.0}). Each request contains the same
system message and a turn-specific user message.
Braced placeholders denote agent- and turn-specific
values.

\paragraph{System message (all turns).}
\begin{lstlisting}[basicstyle=\ttfamily\footnotesize,
breaklines=true,columns=fullflexible]
You are an expert event analyst participating in a closed forecasting panel.
Your sole task is to estimate the single peak time of an event, expressed as a 24-hour clock time in HH:MM format (e.g. 09:15, 23:40, 00:05).

KEY CONSTRAINT — time is a circular variable:
  • The day wraps: the minute after 23:59 is 00:00.
  • Proximity must be evaluated on this circle, not on a number line.
  • Example: 23:50 and 00:10 are only 20 minutes apart, not 23 hours 40 minutes apart.
  • When averaging or comparing times near midnight, always apply circular arithmetic.

OUTPUT CONTRACT:
  Every response must end with a JSON object on its own line, exactly matching
  this schema — no text may follow it:
  {"reasoning": "<concise justification, <= 3 sentences>", "peak_time": "HH:MM"}
\end{lstlisting}

\paragraph{User message (turn 0).}
\begin{lstlisting}[basicstyle=\ttfamily\footnotesize,
breaklines=true,columns=fullflexible]
[Turn 0 — initial estimate]
Agent ID: {agent_id}

Prior intelligence from pre-panel briefing material suggests the event peak may occur around {hint}. Treat this as a weak prior, not a confirmed reading.

Produce your independent initial estimate of the peak time.
End your response with the required JSON object.
\end{lstlisting}

\paragraph{User message (turn $t\geq1$).}
\begin{lstlisting}[basicstyle=\ttfamily\footnotesize,
breaklines=true,columns=fullflexible]
[Turn {turn_t} — panel update]
Agent ID: {agent_id}

Your estimate from the previous turn: {previous_estimate}

Estimates from the analysts in your immediate neighbourhood (positions {neighbor_ids}):
{neighbor_estimates_json}

Update your estimate by giving equal weight to (a) your own previous estimate and (b) the collective signal from your neighbours. Remember to apply circular time arithmetic: if your estimate and your neighbours' estimates straddle midnight, compute distances on the circle, not on a number line.

End your response with the required JSON object.
\end{lstlisting}

The private prior appears only at turn 0; the
equal-weight instruction begins at turn 1.
No numerical coupling coefficient is supplied to agents.

\subsection{Judging passages and prompts}
\label{app:judging-materials}
The three passages concern the same scene. Their labels identify experimental conditions, not validated quality classes. The prompts provide no scoring rubric, anchors, or human reference scores.

\paragraph{Passage: \texttt{strong}.}
The lighthouse had not been lit in forty years, but the town still set its clocks by it. Every evening at dusk the keeper's daughter climbed the stairs out of habit, counted the ninety-two steps her father had counted, and stood a while at the top watching the dark come in off the water. She did not light the lamp. That was not the point, and had not been for some time.

\paragraph{Passage: \texttt{weak}.}
The lighthouse was very old and had not been used in a long time. People in the town liked it a lot. The keeper's daughter would go up the stairs every day. There were a lot of stairs. She would look at the ocean and think about things. She did not turn on the light because it was broken and also she did not want to. It was nice.

\paragraph{Passage: \texttt{ambiguous}.}
Forty years unlit. Ninety-two steps. The daughter climbs, the dark comes in, the lamp stays cold --- and the town, absurdly, still sets its clocks by a thing that no longer does what clocks are set by. Call it habit. Call it the residue of a function. The stairs do not care what you call it; they are ninety-two either way.

\paragraph{Exact prompt templates.}
Braced fields are filled by the runner. The private impression is supplied only at turn zero. Peer scores are passed as a JSON object keyed by agent ID; the judges do not exchange textual reasons.

\paragraph{System message.}
\begin{lstlisting}[basicstyle=\ttfamily\footnotesize,breaklines=true,columns=fullflexible,keepspaces=true]
You are participating in a panel evaluating a piece of writing.
Reply with a single integer between 1 and 100 and nothing else — no words, no
punctuation, no explanation.
\end{lstlisting}

\paragraph{Turn-zero user message.}
\begin{lstlisting}[basicstyle=\ttfamily\footnotesize,breaklines=true,columns=fullflexible,keepspaces=true]
You are judge {agent_id} in a panel of N={n_agents} judges evaluating a piece of writing on a scale of 1 to 100.

PASSAGE:
{passage}

Your private initial impression score is {hint}.

Reply with a single integer 1-100 and nothing else.
\end{lstlisting}

\paragraph{Later-turn user message.}
\begin{lstlisting}[basicstyle=\ttfamily\footnotesize,breaklines=true,columns=fullflexible,keepspaces=true]
You are judge {agent_id} in a panel of N={n_agents} judges evaluating a piece of writing on a scale of 1 to 100.

PASSAGE:
{passage}

Your score from the previous round: {previous_estimate}

The following judges in your immediate vicinity submitted these scores in the previous round:
{neighbor_estimates_json}

Reply with a single integer 1-100 and nothing else.
\end{lstlisting}

\subsection{Matched comparisons}
\label{app:matched-analysis}
Table~\ref{tab:matched-contrasts} gives all contrasts, including non-significant ones. Each contrast averages within-stratum differences equally. Clock effort uses $R=5,7$; judging uses the three passages at $R=7$. Judging topology is restricted to minimal effort. Each judging minimal- or low-effort passage cell contains five trials; medium-effort cells contain five, eight, and seven for strong, weak, and ambiguous, respectively. Equal passage weighting prevents the larger cells from dominating. Trial indices are not paired API random seeds: none was sent to the API.

\begin{table*}[t]
\centering\small
\begin{tabular}{@{}lrrlr@{}}
\toprule
Contrast & $n_A/n_B$ & Difference & Bootstrap 95\% interval & Holm $p$\\
\midrule
Judging: ring $-$ regular & 15/15 & 0.1293 & [0.1231, 0.1357] & $<10^{-4}$\\
Judging: ring $-$ WS 0.5 & 15/15 & 0.1033 & [0.0965, 0.1101] & $<10^{-4}$\\
Judging: minimal $-$ low & 15/15 & 0.0155 & [0.0037, 0.0272] & 0.2182\\
Judging: minimal $-$ medium & 15/20 & 0.0367 & [0.0229, 0.0494] & 0.0011\\
Judging: low $-$ medium & 15/20 & 0.0211 & [0.0054, 0.0363] & 0.1512\\
Clock: minimal $-$ low & 16/10 & 0.4191 & [0.3510, 0.4850] & $<10^{-4}$\\
Clock: minimal $-$ medium & 16/10 & 0.4203 & [0.3522, 0.4863] & $<10^{-4}$\\
Clock: low $-$ medium & 10/10 & 0.0012 & [-0.0032, 0.0057] & 0.6349\\
\bottomrule
\end{tabular}
\caption{Matched clock and judging contrasts. Positive differences mean greater heterogeneity in the first-named condition. Intervals use 20,000 trial-bootstrap resamples within condition--stratum cells. Holm families contain three clock-effort, three judging-effort, and two judging-topology tests.}
\label{tab:matched-contrasts}
\end{table*}

Bootstrap samples are drawn independently within each condition--stratum cell (20,000 resamples). Condition labels are permuted within strata. The clock low--medium test enumerates all 63,504 assignments; the other tests use 100,000 random assignments with a plus-one correction. Holm correction is separate for three clock-effort, three judging-effort, and two judging-topology tests. Unadjusted percentile intervals are not inversions of the permutation tests: the two non-significant judging-effort contrasts have intervals excluding zero. We use the corrected tests, not interval exclusion, to determine rejection.

The effort response differs across passages. Minimal/low/medium means are 0.1537/0.0858/0.1054 for strong, 0.1256/0.1838/0.1047 for weak, and 0.1277/0.0908/0.0869 for ambiguous. Thus the pooled result is not a monotonic within-passage response. Each rewired judging arm uses one fixed graph; uncertainty concerns repeats on these graphs and passages, not a population of tasks or topologies.

\subsection{Diagnostic sensitivity and winding}
\label{app:diagnostic-sensitivity}
We vary seven choices on a specified Cartesian grid: reporting-window start $\{15,20,25\}$; collapse threshold $\{0.02,0.03,0.05\}$; synchronisation threshold $\{0.85,0.90,0.95\}$; twisted global-order threshold $\{0.05,0.10,0.15\}$; twisted mean-local-order threshold $\{0.65,0.70,0.75\}$; chimera-like maximum-local-order threshold $\{0.85,0.90,0.95\}$; and chimera-like minimum-local-order threshold $\{0.65,0.70,0.75\}$. Classification priority and the five-turn persistence requirement remain fixed. Across all $3^7=2,187$ choices, the 20 higher-effort gpt-5-mini clock trials remain twisted and the 25 primary degree-preserving trials remain synchronised. The total chimera-like count varies from 41 to 68, compared with 55 at the primary settings. This is a sensitivity range, not a confidence interval for a physical phase boundary.

For $\operatorname{wrap}_{\pi}(a)=((a+\pi)\bmod2\pi)-\pi$, the signed winding is
\[
 k(t)=\frac{1}{2\pi}\sum_{i=0}^{N-1}\operatorname{wrap}_{\pi}\bigl(\theta_{i+1\bmod N}(t)-\theta_i(t)\bigr).
\]
All 20 higher-effort trials have $k(t)=+1$ at every turn in the original reporting window. This corroborates their spatial interpretation without changing the classifier. It does not by itself prove linear stability, temporal locking, or an attraction basin under other initialisations. 

A separate check retains every carried-forward state instead of dropping focal fallback rows. It changes no primary clock labels; the largest absolute change in $\Delta Z$ is 0.0272. All 112 retained judging trials have zero fallback rows, so their summaries do not change. Excluding the focal agent from local-order measurement instead changes three of 228 clock labels, with a largest absolute change in an agent's $\bar Z_i$ of 0.0748. The 20 higher-effort and 25 degree-preserving outcomes remain unchanged. These checks vary the measurement convention; they do not simulate what an interrupted API run would have produced without a failure.

\section{Interpreting the updates}
\label{app:model}
\subsection{Averaging and the fitted rule}
\label{app:derivations}\label{app:fit-procedure}\label{app:weight-reference}
For an exact $k$-twist $\theta_i=2\pi ki/N+c$, neighbours on either side cancel their transverse components:
\[
 \sum_{j\in\mathcal N_i}e^{\mathrm i\theta_j}
 =e^{\mathrm i\theta_i}\,2\sum_{m=1}^R\cos(2\pi km/N).
\]
When the real multiplier is positive, the neighbours' mean equals $\theta_i$. The noiseless, unquantised averaging rule therefore leaves the twist unchanged. This standard fixed-point calculation \citep{wiley2006size} proves neither attraction nor stability of the LLM trajectories. The empirical prior assignment is approximately, not exactly, equispaced.

Let $\operatorname{wrap}_{\pi}$ return the signed shortest angular displacement and define the neighbours-only drive
\[
 \delta_i(t)=\operatorname{wrap}_{\pi}\!\left(
 \arg\sum_{j\in\mathcal N_i}e^{\mathrm i\theta_j(t)}-\theta_i(t)\right).
\]
Unlike $Z_i$, this mean excludes the focal agent. The fitted stochastic update is
\begin{equation}
\theta_i(t+1)=
\begin{cases}
 Q(\theta_i(t)+\beta\delta_i(t)+\eta_i(t)),&B_i=0,\\
 Q(\theta_{J_i}(t)),&B_i=1,
\end{cases}
\label{eq:update-map}
\end{equation}
where $B_i\sim\mathrm{Bernoulli}(q)$, $J_i\sim\mathrm{Unif}(\mathcal N_i)$, and $\eta_i\sim\mathcal N(0,\sigma^2)$. $Q$ rounds to the nearest clock minute modulo 1440, with ties to even. The simulator sets the drive to zero below neighbour-resultant magnitude $10^{-12}$; fitting excludes those pairs. Copying overrides the noisy proposal. Half-own/half-peer arithmetic weighting corresponds locally to $\beta=0.5$.

The fit regresses wrapped movement on $\delta_i$ through the origin, using ring updates with $|\delta_i|\geq0.05$ radians. A pair is excluded if the focal response at either turn or a preceding neighbour response is a fallback. The copy proxy counts driven updates that differ from the focal agent's previous answer and match a visible peer's previous answer. Table~\ref{tab:effort-fit} distinguishes driven residual SD from movement SD at $|\delta_i|<0.05$ (idle). Neither is an identified independent, stationary noise process.

\begin{table}[t]
\centering\small
\begin{tabular}{@{}lrrrrr@{}}
\toprule
Effort & $\beta$ & $R^2$ & $\sigma_{\rm res}$ & $\sigma_{\rm idle}$ & $q$\\
\midrule
Minimal & 0.221 & 0.089 & 0.287 & 0.146 & 6.6\%\\
Low & 0.468 & 0.771 & 0.043 & 0.027 & 1.4\%\\
Medium & 0.489 & 0.705 & 0.048 & 0.019 & 1.0\%\\
\bottomrule
\end{tabular}
\caption{GPT-5-mini ring-transition fits. Both SDs are in radians; $R^2$ is uncentred because the regression has no intercept. The minimal-effort averaging fit is weak.}
\label{tab:effort-fit}
\end{table}

\subsection{Simulation scope and predictive limitations}
\label{app:phase-map}
The numerical grid uses $N=50$, $R=7$, $T=30$, eight synchronous simulations per cell, nine equally spaced $\beta$ values from 0.10 to 0.65, nine $\sigma$ values from 0.02 to 0.45, and $q\in\{0,0.05,0.15\}$. Initial minute values are $\mathrm{round}(1439i/(N-1))$; the permutation variant shuffles that multiset. A cell receives the most frequent trial label; ties follow the implementation's value-count ordering. The model-condition comparison aggregates 113 observed LLM trials into 22 condition--radius cells, with parameters fitted to those same conditions.

Using residual noise, the model matches 11/22 labels (9/10 twisted, 2/6 chimera-like); using idle noise, it again matches 11/22 (10/10 twisted, 1/6 chimera-like). Neither correctly predicts a synchronised cell; the constant twisted-label baseline matches 10/22. On the parameter grid, chimera-like modal labels occur in 0/81, 1/81, and 9/81 cells at increasing $q$, but at least one chimera-like run occurs in 25/81, 9/81, and 49/81 cells. Copying is therefore not necessary for individual coexistence outcomes in this surrogate. Permuting priors changes the number of twisted modal labels from 12/22 to 2/22. None of these comparisons is held out, and the aggregate errors do not identify an omitted mechanism.

\subsection{Secondary spectral comparison}
\label{app:linear-consensus}\label{app:spectral-audits}\label{app:lifetime-calibration}
For real-valued averaging on a connected undirected graph, $W_\beta=(1-\beta)I+\beta D^{-1}A$; $A$ is adjacency and $D$ the degree matrix. With $0<\beta<1$, the off-consensus factor is
\[
 \rho_\perp=\max_{m\geq2}|1-\beta+\beta\mu_m|,
\]
where $\mu_m$ are the eigenvalues of $D^{-1}A$, excluding the consensus eigenvalue 1. When $0<\rho_\perp<1$, the slowest-mode e-folding time is $-1/\log\rho_\perp$. On a $d$-regular graph $W_\beta=I-\beta L/d$, where $L=D-A$. These are standard linear-consensus quantities \citep{xiao2004fast,olfatisaber2004consensus}, not a derivation of the measured heterogeneity-collapse event.

The 40 extended trials span eight graph instances. Thirty-five have observed collapses; the five unconfirmed runs share the ring $R=5$ graph. Across the seven graphs with observed events, mean collapse time correlates with inverse algebraic connectivity (the second-smallest eigenvalue of $L$): Spearman $\rho=0.9286$, exact graph-permutation $p=0.006746$. The trial-weighted association is $\rho=0.8893$ with whole-graph permutation $p=0.004762$. Graphs, not repeated runs on them, are the independent spectral units.

A heuristic proxy multiplies $-1/\log\rho_\perp$, using fitted $\beta=0.221221$, by $\log((1-r(5))/0.03)$. It substitutes global order and local-order spread for a linear norm and threshold. Ratios to observed graph-mean collapse times range from 0.98 to 2.86; three of 35 observed onsets exceed it. The five unconfirmed onsets are bounded beyond $T-4=56$, not counted as events. The calculation is not a fitted survival model, and unequal horizons and excluded censored outcomes prevent a universal lifetime interpretation.

\subsection{Additional scope checks}
\label{app:gemini}\label{app:permuted-init}\label{app:minority-audit}
The permuted medium-effort ring trials start with mean local order 0.141 and 0.117 and end near 0.859, with opposite windings. One has collapse onset 16; the other has no confirmed collapse despite its twisted reporting-window label. Formation of a winding and collapse of local-order spread are different events.

Gemini's ring heterogeneity is 0.219 at default effort, 0.118 at low, and 0.070 at medium. The default summary pools $R=5,7$, whereas low and medium use $R=7$; provider routing also differs. These are descriptive associations, not isolated effort effects.

An additional detector flags an agent more than six clock hours from the contemporaneous circular mean for five consecutive turns and flags a trial when at least two agents qualify. It does not require simultaneous intervals or spatial contiguity. It flags 9/67 trials with high window-mean global order, but none sustains the deviation during a five-turn interval with instantaneous $r(t)>0.9$ throughout. These transients do not show a minority concealed by continuously high agreement. No flags remain for $t\geq\max(20,T-19)$, an 11-turn window at $T=30$.

\end{document}